\documentclass[a4paper,UKenglish]{lipics-v2016}

\usepackage{microtype}
\usepackage{bm}
\usepackage{cite}
\newcommand{\quotes}[1]{``#1''}
\newtheorem{d1}{Definition}
\newtheorem{r1}{Remark}
\usepackage{algorithm}
\usepackage[noend]{algorithmic}
\newlength\myindent 
\newcommand{\h}{\hspace*{0.2in}}

\title{Lattice Agreement in Message Passing Systems\footnote{Supported by NSF CNS-1563544, NSF CNS-1346245, Huawei Inc., and the Cullen Trust for Higher Education Endowed Professorship.}\footnote{A full version of the paper is available at }}
\titlerunning{Lattice Agreement in Message Passing Systems} 

\author[1]{Xiong Zheng}
\author[2]{Changyong Hu}
\author[3]{Vijay K. Garg}
\affil[1]{University of Texas at Austin, Austin, TX 78712, USA.\\
  \texttt{zhengxiongtym@utexas.edu}}
\affil[2]{University of Texas at Austin, Austin, TX 78712, USA.\\
  \texttt{colinhu9@utexas.edu}}
\affil[3]{University of Texas at Austin, Austin, TX 78712, USA.\\
  \texttt{garg@ece.utexas.edu}}
\authorrunning{X.\,Zheng, C.\,Hu and V.\,K. Garg} 

\Copyright{Xiong Zheng, Changyong Hu and Vijay K. Garg}

\subjclass{Dummy classification -- please refer to \url{http://www.acm.org/about/class/ccs98-html}}

\EventEditors{John Q. Open and Joan R. Acces}
\EventNoEds{2}
\EventLongTitle{42nd Conference on Very Important Topics (CVIT 2016)}
\EventShortTitle{CVIT 2016}
\EventAcronym{CVIT}
\EventYear{2016}
\EventDate{December 24--27, 2016}
\EventLocation{Little Whinging, United Kingdom}
\EventLogo{}
\SeriesVolume{42}
\ArticleNo{23}

\begin{document}

\maketitle

\begin{abstract}
This paper studies the lattice agreement problem and the generalized lattice agreement problem in distributed message passing systems. In the lattice agreement problem, given input values from a lattice, processes have to non-trivially decide output values that lie on a chain. We consider the lattice agreement problem in both synchronous and asynchronous systems. For synchronous lattice agreement, we present two algorithms which run in $\log f$ and $\min \{O(\log^2 h(L)), O(\log^2 f)\}$ rounds, respectively, where $h(L)$ denotes the height of the {\em input sublattice} $L$, $f < n$ is the number of crash failures the system can tolerate, and $n$ is the number of processes in the system. These algorithms have significant better round complexity than previously known algorithms.
The algorithm by Attiya et al. \cite{attiya1995atomic} takes $\log n$ synchronous rounds, and the algorithm by Mavronicolasa \cite{mavronicolasabound} takes $\min \{O(h(L)), O(\sqrt{f})\}$ rounds.
For asynchronous lattice agreement, we propose an algorithm which has time complexity of $2 \cdot \min \{h(L), f + 1\}$ message delays
which improves on the previously known time complexity of $O(n)$ message delays.

The generalized lattice agreement problem defined by Faleiro et al in \cite{faleiro2012generalized} is a generalization of the lattice agreement problem where it is applied for the replicated state machine. We propose an algorithm which guarantees liveness when a majority of the processes are correct in asynchronous systems. Our algorithm requires $\min \{O(h(L)), O(f)\}$ units of time in the worst case which is better than $O(n)$ units of time required by the algorithm of Faleiro et al. \cite{faleiro2012generalized}.
 \end{abstract}

\section{Introduction}
Lattice agreement, introduced in \cite{attiya1995atomic} to solve the atomic snapshot problem \cite{afek1993atomic} in shared memory, is an important decision problem in distributed systems. In this problem, processes start with input values from a lattice and need to decide values which are comparable to each other. Lattice agreement problem is a weaker decision problem than consensus. In synchronous systems, consensus cannot be solved in fewer than $f + 1$ rounds \cite{dolev1983authenticated}, but lattice agreement can be solved in $\log f$ rounds (shown by an algorithm we propose). In asynchronous systems, the consensus problem cannot be solved even with one failure \cite{fischer1985impossibility}, whereas the lattice agreement problem can be solved in asynchronous systems when a majority of processes is correct \cite{faleiro2012generalized}. 

In synchronous message passing systems, a $\log n$ rounds recursive algorithm based on ``branch-and-bound'' approach is proposed in \cite{attiya1995atomic} to solve the lattice agreement problem with message complexity of $O(n^2)$. It can tolerate at most $n - 1$ process failures. Later, \cite{mavronicolasabound} gave an algorithm with round complexity of $\min \{1 + h(L), \lfloor (3 + \sqrt{8f + 1}/2)\rfloor\}$, for any execution where at most $f < n$ processes may crash. Their algorithm has the early-stopping property and is the first algorithm with round complexity that depends on the actual height of the input lattice. Our first algorithm, for synchronous lattice agreement, $LA_{\alpha}$, requires $\log h(L)$ rounds. It assumes that the height of the input lattice is known to all processes. By applying this algorithm as a building block, we give an algorithm, $LA_{\beta}$, which requires only $\log f$ rounds without the height assumption in $LA_{\alpha}$. Instead of directly trying to decide on the comparable output values which are from a lattice with an unknown height, this algorithm first performs lattice agreement on the failure set known by each process by using $LA_{\alpha}$. Then each process removes values from \textit{faulty} processes they know and outputs the join of all the remaining values. Our third algorithm, $LA_{\gamma}$, has round complexity of $\min \{O(\log^2 h(L)), O(\log^2 f))$, which depends on the height of the input lattice but does not assume that the height is known. This algorithm iteratively guesses the actual height of the input lattice and applies $LA_{\alpha}$ with the guessed height as input, until all processes terminate. 

Lattice agreement in asynchronous message passing systems is useful due to its applications in
atomic snapshot objects and fault-tolerant replicated state machines. 
Efficient implementation of atomic snapshot objects in crash-prone asynchronous message passing systems is important because they can make design of algorithms in such systems easier (examples of algorithms in message passing systems based on snapshot objects  can be found in \cite{taubenfeld2006synchronization},\cite{raynal2012concurrent} and \cite{attiya2004distributed}). As shown in \cite{attiya1995atomic}, any algorithm for lattice agreement can be applied to solve the atomic snapshot problem in a shared memory system. We note that \cite{attiya1998atomic} does not directly use lattice agreement to solve the atomic snapshot problem, but their idea of producing comparable \textit{views} for processes is essentially lattice agreement. Thus, by using the same transformation techniques in \cite{attiya1995atomic} and \cite{attiya1998atomic}, algorithms for lattice agreement problem can be directly applied to implement atomic snapshot objects in crash-prone message passing systems. We give an algorithm for asynchronous lattice agreement problem which requires $\min \{O(h(L)), O(f)\}$ message delays. Then, by applying the technique in \cite{attiya1998atomic}, our algorithm can be used to implement atomic snapshot objects on top of crash-prone asynchronous message passing systems and achieve time complexity of $O(f)$ message delays in the worst case. 
Our result significantly improves the message delays in the previous work by Delporte-Gallet, Fauconnier et al\cite{delporte2016implementing}. The algorithm in \cite{delporte2016implementing} directly implements an atomic snapshot object on top of crash-prone message passing systems and requires $O(n)$ message delays in the worst case. 

Another related work for lattice agreement in asynchronous systems is by Faleiro et al. \cite{faleiro2012generalized}. They solve the lattice agreement problem in asynchronous systems by giving a Paxos style protocol \cite{lamport1998part,lamport2001paxos}, in which each proposer keeps proposing a value until it gets \textit{accept} from a majority of acceptors. The acceptor only \textit{accepts} a proposal when the proposal has a bigger value than its accepted value. Their algorithm requires $O(n)$ message delays. Our asynchronous lattice agreement algorithm does not have Paxos style. Instead, it runs in \textit{round-trips}. Each round-trip is composed of sending a message to all and getting $n - f$ acknowledgements back. Our algorithm guarantees termination in $\min \{O(h(L)), O(f)\}$ message delays which is a significant improvement over $O(n)$ message delays. 

Generalized lattice agreement problem defined in \cite{faleiro2012generalized} is a generalization of the lattice agreement problem in asynchronous systems. It is applied to implement a specific class of replicated state machines. In conventional replicated state machine approach \cite{schneider1990implementing}, consensus based mechanism is used to implement strong consistency.
Due to performance reasons, many systems relax the strong consistency requirement and support eventual consistency \cite{tanenbaum2007distributed}, i.e, all copies are eventually consistent. However, there is no guarantee on when this eventual consistency happens. Also, different copies could be in an inconsistent state before this eventual situation happens.  Conflict-free replicated data types (CRDT) \cite{shapiro2011conflict,shapiro2011convergent} is a data structure which supports such eventual consistency. In CRDT, all operations are designed to be commutative such that they can be concurrently executed without coordination. As shown in \cite{faleiro2012generalized}  by applying generalized lattice agreement on top of CRDT, the states of any two copies can be made comparable and thus provide linearizability guarantee \cite{herlihy1990linearizability} for CRDT. 

The following example from \cite{faleiro2012generalized}
motivates generalized lattice agreement. Consider a replicated set data structure which supports adds and reads. Suppose there are two concurrent updates, \textit{add(a)} and \textit{add(b)}, and two concurrent reads on copy one and two respectively. By using CRDT, it could happen that the two reads return \textit{\{a\}} and \textit{\{b\}} respectively. This execution is not linearizable \cite{herlihy1990linearizability}, because if {\em add(a)} appears before {\em add(b)} in the linear order, then no read can return {\em \{b\}}. 
On the other hand, if we use conventional consensus replicated state machine technique, then all operations would be coordinated including the two reads. 
This greatly impacts the throughput of the system. 
By applying generalized lattice agreement on top of CRDT, all operations can be concurrently executed and any two reads always return comparable views of the system. In the above example, the two reads return either (i) \textit{\{a\}} and \textit{\{a, b\}} or (ii) \textit{\{b\}} and \textit{\{a, b\}} which is linearizable. Therefore, generalized lattice agreement can be applied on top of CRDT to provide better consistency guarantee than CRDT and better availability than conventional replicated state machine technique.

Since the generalized lattice agreement problem has applications in building replicated state machines, it
is important to reduce the message delays for a value to be learned. Faleiro et al. \cite{faleiro2012generalized} propose an algorithm for the generalized lattice agreement by using their algorithm for the lattice agreement problem as a building block. Their generalized lattice agreement algorithm satisfies safety and liveness assuming $f < \lceil \frac{n}{2} \rceil$. A value is eventually learned in their algorithm after $O(n)$ message delays in the worst case.
Our algorithm guarantees that a value is learned in $\min \{O(h(L)), O(f)\}$ message delays.

In summary, this paper makes the following contributions:
\begin{itemize}
\item We present an algorithm, $LA_\alpha$ to solve the lattice agreement in synchronous system in $\log h(L)$ rounds assuming $h(L)$ is known. Using $LA_\alpha$, we propose an algorithm, $LA_\beta$ to solve the standard lattice agreement problem in $\log f$ rounds. 
This bound is significantly better than the previously known upper bounds of $\log n$ by \cite{attiya1998atomic} and min$\{1 + h(L), \lfloor (3 + \sqrt{8f + 1}/2)\rfloor\}$ by \cite{mavronicolasabound} (and solves the open problem posed there). We also give an algorithm, $LA_\gamma$ which runs in $\min \{O(\log^2 h(L)), O(\log^2 f)\}$ rounds. 
\item For the lattice agreement problem in asynchronous systems, we  give an algorithm, $LA_\delta$ which requires $2 \cdot \min \{h(L), f + 1\}$ message delays which improves the $O(n)$ bound by \cite{faleiro2012generalized}.
\item Based on the asynchronous lattice agreement algorithm, we present an algorithm, $GLA_\alpha$, to solve the generalized lattice agreement with time complexity $\min \{O(h(L)), O(f)\}$ message delays which improves the $O(n)$ bound by \cite{faleiro2012generalized}. 
\end{itemize}

Related previous work and our results are summarized in Table \ref{tab:related_work}. {\em LA sync} and {\em LA async} represent lattice agreement in synchronous systems and asynchronous systems, respectively. {\em GLA async} represents generalized lattice agreement in asynchronous systems.  $LA_{\alpha}$ is designed to solve the lattice agreement problem with the assumption that the height of the input lattice is given. It serves as a building block for $LA_{\beta}$ and $LA_{\gamma}$. For synchronous systems, the time complexity is given in terms of synchronous rounds. For asynchronous system, the time complexity is given in terms of message delays. The message column represents the total number of messages sent by all processes in one execution. For generalized lattice agreement problem, the message complexity is given in terms of the number of messages needed for a value to be learned. 

\begin{table*}[h]
\caption{Previous Work and Our Results \label{tab:related_work}}
\centering
\begin{tabular}{ |c|c|c|c|} 
 \hline
 \textbf{Problem} & \textbf{Protocol} &  \textbf{Time} & \textbf{Message} \\ \hline
 \multirow{5}{*}{LA sync} & [3] &  $O(\log n)$ & $O(n^2)$ \\   
 \cline{2-4} 
 & \cite{mavronicolasabound} & $\min \{O(h(L)), O(\sqrt{f})\}$ & $n^2 \cdot \min \{O(h(L)), O(\sqrt{f}) \}$\\ 
 \cline{2-4} 
 & $LA_{\alpha}$ & $O(\log h(L))$ & $O(n^2 \log h(L))$ \\
 \cline{2-4} 
 & $LA_{\beta}$ & $O(\log f)$ & $O(n^2 \log f)$ \\ 
 \cline{2-4} 
 & $LA_{\gamma}$ & $\min \{O(\log^2 h(L)), O(\log^2 f)\}$ & $n^2 \cdot \min \{O(\log^2 h(L)), O(\log^2 f)\}$\\ 
 \cline{1-4} 
\multirow{2}{*}{LA async} & \cite{faleiro2012generalized}  & $O(n)$& $O(n^3)$\\
 \cline{2-4} 
 & $LA_{\delta}$   & $\min \{O(h(L)), O(f)\}$& $ n^2 \cdot \min \{O(h(L)), O({f}) \}$\\
 \cline{1-4}
 \multirow{2}{*}{GLA async} & \cite{faleiro2012generalized} & $O(n)$ & $O(n^3)$ \\ 
 \cline{2-4}  
 & $GLA_{\alpha}$ &   $\min \{O(h(L)), O(f)\}$ & $n^2 \cdot \min \{O(h(L)), O(f)\}$ \\
 \hline
 
\end{tabular}
\vspace*{-0.2in}
\end{table*}

\section{System Model and Problem Definitions}
\subsection{System Model}
We assume a distributed message passing system with $n$ processes in a completely connected topology, denoted as $p_1,...,p_n$. We consider both synchronous or asynchronous systems. Synchronous means that message delays and the duration of the operations performed by the process have an upper bound on the time. Asynchronous means that there is no upper bound on the time for a message to reach its destination. The model assumes that processes may have crash failures but no Byzantine failures. 
The model parameter $f$ denotes the maximum number of processes that may crash in a run. We assume that the underlying communication system is reliable but the message channel may not be FIFO. We say a process is \textit{faulty} in a run if it crashes and \textit{correct} or \textit{non-faulty} otherwise. In our following algorithms, when a process sends a message to all, it also sends this message to itself. 

\subsection{Lattice Agreement}
Let ($X$, $\leq$, $\sqcup$) be a finite join semi-lattice with a partial order $\leq$ and join $\sqcup$. Two values $u$ and $v$ in $X$ are comparable iff $u \leq v$ or $v \leq u$. The join of $u$ and $v$ is denoted as $\sqcup \{u, v\}$. $X$ is a \textit{join semi-lattice} if a join exists for every nonempty finite subset of $X$. As customary in this area, we use the term {\em lattice} instead of
{\em join semi-lattice} in this paper for simplicity.  

In the lattice agreement problem \cite{attiya1995atomic}, each process $p_i$ can propose a value $x_i$ in $X$ and must decide on some output $y_i$ also in $X$. An algorithm is said to solve the lattice agreement problem if the following properties are satisfied: 

\textbf{Downward-Validity}: For all $i \in [1..n]$, $x_i \leq y_i$. 

\textbf{Upward-Validity}: For all $i \in [1..n]$, $y_i \leq \sqcup\{x_1,...,x_n\}$.

\textbf{Comparability}: For all $i \in [1..n]$ and $j 
\in [1..n]$, either $y_i \leq y_j$ or $y_j \leq y_i$.

In this paper, all the algorithms that we propose apply join operation to some subset of input values.
Therefore, it is sufficient to focus on the join-closed subset of $X$ that includes all
input values. Let $L$ be the join-closed subset of $X$ that includes all input values. $L$ is also a join semi-lattice. We call $L$ the {\em input sublattice} of $X$. All algorithms proposed in this paper are based on $L$. Since the complexity of our algorithms depend on the height of lattice $L$, we give the formal definitions as below:

\begin{d1}
The height of a value $v$ in a lattice $X$ is the length of longest path from any minimal value to $v$, denoted as $h_X(v)$ or $h(v)$ when it is clear. 
\end{d1}

\begin{d1}
The height of a lattice $X$ is the height of its largest value, denoted as $h(X)$.
\end{d1}

Each process proposes a value and they form a boolean lattice. Thus, the largest value in this lattice is the union of all values which has size of $n$. Therefore, from the definition 2, we have $h(L) \leq n$. 

\subsection{Generalized Lattice Agreement}
In generalized lattice agreement problem, each process may receive a possibly infinite sequence of values as inputs that belong to a lattice at any point of time. Let $x_i^{p}$ denote the $i$th value received by process $p$. The aim is for each process $p$ to learn a sequence of output values $y_j^p$ which satisfies the following conditions:

\textbf{Validity}: any learned value $y_j^p$ is a join of some set of received input values.

\textbf{Stability}: The value learned by any process $p$ is non-decreasing: $j < k \implies y_j^{p} \leq y_k^p$.

\textbf{Comparability}: Any two values $y_j^p$ and $y_k^q$ learned by any two process $p$ and $q$ are comparable. 

\textbf{Liveness}: Every value $x_i^p$ received by a correct process $p$ is eventually included in some learned value $y_k^q$ of every correct process $q$: i.e, $x_i^p \leq y_k^q$.

\section{Lattice Agreement in Synchronous Systems}
\subsection{Lattice Agreement with Known Height}
In this section, we first consider a simpler version of the standard lattice agreement problem by assuming that the height of the input sublattice $L$ is known in advance, i.e, $h(L)$ is given. We propose an algorithm, $LA_{\alpha}$, to solve this problem in $\log h(L)$ synchronous rounds. In section 3.2, we give an algorithm to solve the lattice agreement problem when the height is not given using this algorithm.

Algorithm $LA_{\alpha}$ runs in synchronous rounds. At each round, by calling a \textit{Classifier} procedure (described below), processes within a same group (to be defined later) are classified into different groups.
The algorithm guarantees that any two processes within the same group have equal values and any two processes in different groups have comparable values at the end. Thus, values of all processes are comparable to each other at the end. We present the algorithm by first introducing the fundamental \textit{Classifier} procedure.

\subsubsection{The Classifier Procedure}
The \textit{Classifier} procedure is inspired by the Classifier procedure given by Attiya and Rachman in \cite{attiya1998atomic}, called \textit{AR-Clasifier}, where it is applied to solve the atomic snapshot problem in the shared memory system. The intuition behind the \textit{Classifier} procedure is to classify processes to \textit{master} or \textit{slave} and ensure all \textit{master} processes have values greater than all \textit{slave} processes. 
 
The pseudo-code for \textit{Classifier} is given in Figure \ref{fig:classifier}. It takes two parameters: the input value $v$ and the threshold value $k$. The output is composed of three items: the output value, the classification result and the decision status. The process which calls the \textit{Classifier} procedure should update their value to be the output value. The classification result is either \textit{master} or \textit{slave}. The decision status is a Boolean value which is used to inform whether the invoking process can decide on the output value or not. The main functionality of the \textit{Classifier} procedure is either to tell the invoking process to decide, or to classify the invoking process as a \textit{master} or a \textit{slave}. Details of the \textit{Classifier} procedure are shown below:\\
\h Line 1-3: The invoking process sends a message with its input value $v$ and the threshold value $k$ to all. It then collects all the received values associated with the threshold value $k$ in a set $U$.\\
\h Line 5-6: It checks whether all values in $U$ are comparable to the input value. If they are comparable, it terminates the \textit{Classifier} procedure and returns the input value as the output value and \textit{true} as the decision status. \\
\h Line 8-12: It performs classification based on received values. Let $w$ be the join of all received values associated with the threshold value $k$. If the height of $w$ in lattice $L$ is greater than the threshold value $k$, then the \textit{Classifier} returns $w$ as the output value, \textit{master} as the classification result and false as the decision status. Otherwise, it returns the input value as the output value, \textit{slave} as the classification result and false as the decision status. From the classification steps, it is easy to see that the processes classified as \textit{master} have values greater than those classified as \textit{slave} because $w$ is the join of all values in $U$.

There are four main differences between the \textit{AR-Classifier} and our {\em Classifier}: 1) The \textit{AR-Classifier} is based on the shared memory model whereas our algorithm is based on synchronous message passing. 2) The \textit{AR-Classifier} does not allow early termination. 3) Each process in the \textit{AR-Classifier} needs values from all processes whereas our \textit{Classifier} uses values only from processes within its group. 4) The \textit{AR-Classifier} procedure requires the invoking process to read values of all processes again if the invoking process is classified as \textit{master} where as our algorithm needs to receive values from all processes only once.

\begin{figure}[htb]
\fbox{\begin{minipage}[t]  {3.1in}
\noindent
\underline{{\bf \textit{Classifier}$(v,k)$:}} \\
$v$: input value \h $k$: threshold value 

\begin{algorithmic}[1]
\STATE Send $(v, k)$ to all
\STATE Receive messages of the form $(-, k)$ 
\STATE Let $U$ be values contained in received messages
\item
\STATE /* Early  Termination */
\STATE {\bf if} {$|U| = 0$ or  $\forall u \in U: v \leq u \vee u \leq v$}  
\STATE \h {\bf return} ($v$, $-$, $true$)
\item
\STATE /* Classification */
\STATE Let $w := \sqcup\{u: u \in U\}$
\STATE {\bf if} {$h(w) > k$} 
\STATE \h {\bf return} ($w$, \textit{master}, $false$)
\STATE {\bf else}
\STATE \h {\bf return} ($v$, \textit{slave}, $false$)
\end{algorithmic}
\caption{\textit{Classifier} \label{fig:classifier}}
\end{minipage}
} 
\fbox{\begin{minipage}[t]  {2.1in}
\noindent
\underline{$\bm{LA_{\alpha}}(H, x_i)$ for $p_i$:} \\
$H$: given height \h  $x_i$: input value 

\begin{algorithmic}[1]
\STATE $v_i^{1} := x_i$ // value at round 1 
\STATE $l_i := \frac{H}{2}$ // label
\STATE $decided := false$
\STATE 
\STATE {\bf for} {$r := 1$ to $\log H + 1$}
\STATE \h ($v_i^{r + 1}, class, decided$) \\
\h\h$:= Classifier(l_i, v_i^{r})$
\STATE \h {\bf if} {\textit{decided}} 
\STATE \h \h {\bf return} $v_i^{r + 1}$
\STATE \h {\bf else if} {$class = master$} 
\STATE \h \h $l_i := l_i + \frac{H}{2^{r + 1}}$
\STATE \h {\bf else} 
\STATE \h \h $l_i := l_i - \frac{H}{2^{r + 1}}$
\STATE {\bf end for}
\end{algorithmic}
\caption{Algorithm $LA_{\alpha}$ \label{fig:la_alpha}}
\end{minipage}
} 
\vspace*{-.2in}
\end{figure}

\subsubsection{Algorithm $LA_{\alpha}$}
Algorithm $LA_{\alpha}$ (shown in Figure \ref{fig:la_alpha}) runs in at most $\log h(L)$ rounds. It assumes knowledge of $H = h(L)$, the height of the input lattice. Let $x_i$ denote the initial input value of process $i$, $v_i^{r}$ denote the value held by process $i$ at the beginning of round $r$, and $class$ denote the classification result of the \textit{Classifier} procedure. The $class$ indicates whether the process is classified as a \textit{master} or a \textit{slave}. The \textit{decided} variable shows whether the process has decided or not. Each process $i$ has a label denoted as $l_i$. This label is updated at each round. Processes which have the same label $l$ are said to be in the same group with label $l$.
The definitions of $label$ and $group$ are formally given as:

\begin{d1}[label]
Each process has a $label$, which serves as a knowledge threshold and is passed as the threshold value $k$ whenever the process calls the \textit{Classifier} procedure.
\end{d1}

\begin{d1}[group]
A $group$ is a set of processes which have the same label. The label of a group is the label of the processes in this group.
\end{d1}

A process has {\em decided} if it has set its decision status to true. Otherwise, it is undecided. At each round $r$, an undecided process invokes the \textit{Classifier} procedure with its current value and its current label $l_i$ as parameters $v$ and $k$, respectively. Since each process passes its label as the threshold value $k$ when invoking the \textit{Classifier} procedure, \textit{line 2} of the \textit{Classifier} is equivalent to receiving messages from processes within the same group; that is, at each round, a process performs the \textit{Classifier} procedure within its group. Processes which are in different groups do not affect each other. At round $r$, by invoking the \textit{Classifier} procedure, each process $i$ sets $v_i^{r + 1}$, $class$ and \textit{decided} to the returned output value, the classification result and the decision status. Each process first checks the value of \textit{decided}. If it is true, process $i$ decides on $v_i^{r + 1}$ and terminates the algorithm. Otherwise, if it is classified as a \textit{master}, it increases its label by $\frac{H}{2^{r + 1}}$. If it is classified as a \textit{slave}, it decreases its label by $\frac{H}{2^{r + 1}}$. Now we show how the \textit{Classifier} procedure combined with this label update mechanism makes any two processes have comparable values at the end. 

 Let $G$ be a group of processes at round $r$. Let $M(G)$ and $S(G)$ be the group of processes which are classified as \textit{master} and \textit{slave}, respectively, when they run the \textit{Classifier} procedure in group $G$. We say that $G$ is the parent of $M(G)$ and $S(G)$. Thus, $M(G)$ and $S(G)$ are both groups at round $r + 1$. Process $i \in M(G)$ or $i \in S(G)$ indicates that $i$ does not decide in group $G$ at round $r$. Initially, all process have the same label $\frac{H}{2}$ and are in the same group with label $\frac{H}{2}$. When they execute the \textit{Classifier}, they will be classified into different groups. We can view the execution as processes traversing through a binary tree. Initially, all of them are at the root of the tree. As the program executes, if they are classified as \textit{master}, then they go to the right child. Otherwise, they go to the left child. 

Before we prove the correctness of the given algorithm, we first give some useful properties satisfied by the \textit{Classifier} procedure. Although Lemma \ref{lem:cls} is similar to a lemma given in [5], it is discussed here in message passing systems and the proofs are different. 

\begin{lemma}\label{lem:cls}
Let $G$ be a group at round $r$ with label $k$. Let $L$ and $R$ be two nonnegative integers such that $L \leq k \leq R$. If $L < h(v_i^{r}) \leq R$ for every process $i \in G$, and $h(\sqcup\{v_i^{r}: i \in G\}) \leq R$, then \\
(p1) for each process $i \in M(G)$, $k < h(v_i^{r + 1}) \leq R$\\
(p2) for each process $i \in S(G)$, $L < h(v_i^{r + 1}) \leq k$\\
(p3) $h(\sqcup\{v_i^{r + 1}: i \in M(G)\}) \leq R$ \\
(p4) $h(\sqcup\{v_i^{r + 1}: i \in S(G)\}) \leq k$, and\\
(p5) for each process $i \in M(G)$, $v_i^{r + 1} \geq \sqcup\{v_i^{r + 1}: i \in S(G)\}$   
\begin{proof}
\textbf{(p1)-(p3):} Immediate from the \textit{Classifier} procedure.\\
\textbf{(p4):}  Since $S(G)$ is a group of processes which are at round $r + 1$, all processes in $S(G)$ are correct (non-faulty) at round $r$. So, all processes in $S(G)$ must have received values of each other in the \textit{Classifier} procedure at round $r$ in group $G$. Thus, $h(\sqcup\{v_i^{r + 1}: i \in S(G)\}) \leq k$, otherwise all of them should be in  group $M(G)$ instead of $S(G)$, according to the condition at \textit{line 9} of the \textit{Classifier} procedure. \\
\textbf{(p5):} Since all processes in $S(G)$ are correct at round $r$, all processes in $M(G)$ must have received values of all processes in $S(G)$ in the \textit{Classifier} procedure at round $r$. Any process which proceeds to group $M(G)$ takes the join of all received values at round $r$, according to \textit{line 10}. Thus, for every process $i \in M(G)$, $v_i^{r + 1} \geq \sqcup\{v_i^{r + 1}: i \in S(G)\}$.
\end{proof}
\end{lemma}

\begin{lemma}\label{lem:lat}
Let $x$ be a value from a lattice $L$, and $V$ be a set of values from $L$. Let $U$ be any subset of $V$. If $x$ is comparable with $\forall v \in V$, then $x$ is comparable with $\sqcup\{u~ |~ u \in U\}$.  
\begin{proof}
If $\forall u \in U: u \leq x$, then $\sqcup\{u ~| ~u \in U\} \leq x$. Otherwise, $\exists y \in U: x \leq y$. Since $y \leq \sqcup\{u ~|~ u \in U\}$, so $x \leq \sqcup\{u ~|~ u \in U\}$.
\end{proof}
\end{lemma}

\begin{lemma}\label{lem:decideComp}
If process $i$ decides at round $r$ on value $y_i$, then $y_i$ is comparable with $v_j^r$ for any correct process $j$. 
\begin{proof}
Let process $i$ decide in group $G$ at round $r$. Consider the two cases below:

Case 1: $j \not \in G$. Let $G'$ be a group at the maximum round $r'$ such that both $i$ and $j$ belong to $G'$. Then, either $i \in M(G') \wedge j \in S(G')$ or $j \in M(G') \wedge i \in S(G')$. We only consider the case $i \in M(G') \wedge j \in S(G')$. The other case can be proved similarly. From ($p5$) of Lemma \ref{lem:cls}, we have $\sqcup \{v_p^{r}: p \in S(G')\} \leq y_i$. Since $j \in S(G')$, then $v_j^r \leq \sqcup \{v_p^{r}: p \in S(G')\}$. Thus, $v_j^r \leq y_i$. For the other case, we have $y_i \leq v_j^r$. Therefore, $y_i$ is comparable with $v_j^r$. 

Case 2: $j \in G$, since process $j$ is correct, then $i$ must have received $v_j^r$ at round $r$. Thus, by {\em line 5} of the \textit{Classifier} procedure, we have that $y_i$ is comparable with $y_j^r$.  
\end{proof}
\end{lemma}
Now we show that any two processes decide on comparable values. 
\begin{lemma}\label{lem:comp_alpha}
(Comparability) Let process $i$ and $j$ decide on $y_i$ and $y_j$, respectively. Then $y_i$ and $y_j$ are comparable.
\begin{proof}
Let process $i$ and $j$ decide at round $r_i$ and $r_j$, respectively. Without loss of generality, assume $r_i \leq r_j$. At round $r_i$, from Lemma \ref{lem:decideComp} we have $y_i$ is comparable with $v_k^r$ for any correct undecided process $k$. Let $V = \{v_k^{r_i} ~|~ process ~ k~ undecided ~ and ~correct\}$. Since $r_j \geq r_i$, $y_j$ is at most the join of a subset of $V$. Thus, from Lemma \ref{lem:lat} we have $y_i$ and $y_j$ are comparable.
\end{proof}
\end{lemma}

Now we prove that all processes decide within $\log H + 1$ rounds by showing all processes in the same group at the beginning of round $\log H + 1$ have equal values, given by Lemma \ref{lem:dec} and Lemma \ref{lem:term}. Since Lemma \ref{lem:dec} and Lemma \ref{lem:term} and the corresponding proofs are similar to the ones given in \cite{attiya1998atomic}, the proofs are omitted here and can be found in the full paper. Proof of Lemma \ref{lem:dec} is based on ($p1$-$p4$) of Lemma \ref{lem:cls} by induction. Proof of Lemma  \ref{lem:term} is based on Lemma \ref{lem:dec}. 

\begin{lemma}\label{lem:dec}
Let $G$ be a group of processes at round $r$ with label $k$. Then \\
(1) for each process $i \in G$, $k - \frac{H}{2^r} < h(v_i^r) \leq k + \frac{H}{2^r}$ \\
(2) $h(\sqcup\{v_i^r: i \in G\}) \leq k + \frac{H}{2^r}$




\end{lemma}

\begin{lemma}\label{lem:term}
Let $i$ and $j$ be two processes that are within the same group $G$ at the beginning of round $r = \log H + 1$. Then $v_i^{r}$ and $v_j^{r}$ are equal.

\end{lemma}

\begin{lemma}\label{lem:decide}
All processes decide within $\log H + 1$ rounds.
\begin{proof}
From Lemma \ref{lem:term}, we know any two processes which are in the same group at the beginning of round $\log H + 1$ have equal values. Then, the condition in \textit{line 5} of \textit{Classifier} procedure is satisfied. Thus, all undecided processes decide at round $\log H + 1$.
\end{proof}
\end{lemma}

\begin{r1}
Since at the beginning of round $\log H + 1$ all undecided processes have comparable values, $LA_{\alpha}$ only needs $\log H$ rounds. For simplicity, one more round is executed to make all processes decide at \textit{line 5} of the \textit{Classifier} procedure. 
\end{r1}

\begin{theorem}\label{theo:comp_alpha}
Algorithm $LA_{\alpha}$ solves lattice agreement problem in $\log H$ rounds and can tolerate $f < n$ failures.
\begin{proof} \textit{Downward-Validity} follows from the fact that the value of each process is non-decreasing at each round. For \textit{Upward-Validity}, according to the \textit{Classifier} procedure, each process either keeps its value unchanged or takes the join of the values proposed by other processes which could never be greater than $\sqcup\{x_1,...,x_n\}$. For \textit{Comparability}, from Lemma \ref{lem:comp_alpha}, we know for any two process $i$ and $j$, if they decide, then their decision values must be comparable. From Lemma \ref{lem:decide}, we know all processes decide. Thus, comparability holds. 
\end{proof}
\end{theorem}

{\bf Complexity}. Time complexity is $\log H$ rounds. For message complexity, since each process sends $n$ messages per round, $\log H$ rounds results in $n^2\log H$ messages in total. Notice that the number of messages can be further reduced by keeping a set of processes which are not in its group. If a process $p$ receives a message from process $q$ with a threshold value different from its own threshold value, it knows that $q$ is not in its group.
Each process does not send messages to the processes in this set. 

Algorithm $LA_{\beta}$ runs in $\log height(L)$ rounds by assuming that $height(L)$ is given. However, in order to know that actual height of input lattice, we need to know how many distinc values all process propose which needs extra efforts. For this reason, in following sections, we introduce algorithms to solve the lattice agreement problem without this assumption, based on algorithm $LA_{\beta}$.

\subsection{Lattice Agreement with Unknown Height}
In this section, we consider the standard lattice agreement in which the height of the lattice is not known to any process. We propose algorithm, $LA_{\beta}$, (shown in Figure \ref{fig:la_beta}) based on algorithm $LA_{\alpha}$. 

\subsubsection{Algorithm $LA_{\beta}$}
Algorithm $LA_{\beta}$ runs in $\log f + 1$ synchronous rounds. It makes use of algorithm $LA_{\alpha}$ as a building block. Instead of directly agreeing on input values which are taken from a lattice with unknown height, we first do lattice agreement on the failure set that each process knows {\em after one round of broadcast}. The set of all failure sets forms a boolean lattice with union be the join operation and with height of $f$, since there are at most $f$ failures. The algorithm consists of two phases. At \textit{Phase A}, all processes exchange their values. Process $i$ includes $j$ into its failure set if it does not receive value from process $j$ at the first phase. After the first phase, each process has a failure set which contains failed processes it knows. Then in \textit{phase B}, they invoke algorithm $LA_{\alpha}$ with $f$ as the height and its failure set as input. After that, each process {\em decides} on a failure set which satisfies lattice agreement properties. The new failure set of any two process $i$ and $j$ are comparable to each other, i.e, $F_i^{'}$ is comparable to $F_j^{'}$. Equipped with this comparable failure set, each process removes values it received from processes which are in its failure set and decides on the join of the remaining values. 
 
\begin{figure}[htb]
\fbox{\begin{minipage}[t]  {3.0in}
\noindent
\underline{$\bm{LA_{\beta}}$ for $p_i$}
\begin{algorithmic}[1]
\STATE $V_i := \{x_i\}$ // set of values, initially $x_i$
\STATE $F_i := \emptyset$ // set of known failure processes  
\STATE $f :=$ the maximum number of failures 
\item
\STATE /* Exchange Values and Record Failures*/
\STATE \textbf{\textit{Phase A:}}
\STATE Send $V_i$ to all 
\STATE {\bf for} {$j := 1$ to $n$}
\STATE \h{\bf if} {$V_j$ is received from process $j$}
\STATE \h\h $V_i := V_i \cup {V_j}$
\STATE \h{\bf else} 
\STATE \h\h $F_i := F_i \cup {j}$
\STATE {\bf end for}
\item 
\STATE /* LA with Known Height $f$ */
\STATE \textbf{\textit{Phase B:}}
\STATE $F_i^{'} := LA_{\alpha}(f, F_i)$
\STATE $U_i :=$ values from processes in $F_i^{'}$ in \textit{Phase A}
\STATE $C_i := V_i - U_i$  // set of correct values  
\STATE $y_i := \sqcup\{v: v \in C_i\}$
\end{algorithmic}
\caption{Algorithm $LA_{\beta}$ \label{fig:la_beta}}
\end{minipage}
} 
\fbox{\begin{minipage}[t]  {2.2in}
\underline{$\bm{LA_{\gamma}}$ for $p_i$}

\begin{algorithmic}[1]
\STATE $v_i := x_i$ // input value
\STATE $decided := false$
\item
\STATE /* Exchange Values */
\STATE \textbf{\textit{Phase A:}}
\STATE Send $v_i$ to all 
\STATE {\bf for} {$j := 1$ to $n$} 
\STATE \h receive $v_j$ from $p_j$
\STATE \h $v_i := v_i \sqcup v_j$
\STATE {\bf end for}
\item 
\STATE /* Guessing Height */
\STATE \textbf{\textit{Phase B:}}
\STATE $guess := 2$ // guess height
\STATE {\bf while} {($! decided$)} 
\STATE \h $v_i := LA_{\alpha}(guess, v_i)$
\STATE \h $guess := 2*guess$
\STATE {\bf end while}
\item 
\STATE $y_i := v_i$
\end{algorithmic}
\caption{Algorithm $LA_{\gamma}$ \label{fig:la_gamma}}
\end{minipage}
} 
\vspace*{-0.2in}
\end{figure}
The following lemma shows that any two processes decide on comparable values. We only give the sketch of proof, and detailed proof is available in the full paper. 
\begin{lemma}\label{lem:comp_beta}
(Comparability) Let process $i$ and $j$ decide on $y_i$ and $y_j$, respectively. Then $y_i$ and $y_j$ are comparable.
\begin{proof}
(Sketch of proof) According to comparability of $LA_{\alpha}$, all processes have comparable failure sets. Then, the set of values they received at \textit{Phase A} from correct processes must be comparable, i.e, $C_i$ is comparable with $C_j$. Therefore, $y_i$ and $y_j$ are comparable.
\end{proof}
\end{lemma}

\begin{theorem}
$LA_{\beta}$ solves lattice agreement problem in $\log f + 1$ rounds, where $f < n$ is the maximum number of failures the system can tolerate.
\begin{proof} \textit{Downward-Validity}. Initially, for correct process $i$, $v_i = x_i$. After \textit{Phase A}, since $i$ is correct, so $i$ is not in any failure set of any process. At \textit{Phase B}, process $i$ invokes algorithm $LA_{\alpha}$ with failure set as the input value. Thus, according to the \textit{Upward-Validity} of $LA_{\alpha}$, $i$ is not included in $F_i^{'}$. So, $x_i \in C_i$. Therefore, $x_i \leq y_i$. \textit{Upward-Validity} is immediate from the fact that each process receives at most all values by all processes. \textit{Comparability} follows from Lemma \ref{lem:comp_beta}.
\end{proof}
\end{theorem}

\subsubsection{Algorithm $LA_{\gamma}$}
Algorithm $LA_{\beta}$ solves lattice agreement in $\log f + 1$ rounds whereas Algorithm $LA_{\alpha}$ solves lattice agreement in $\log h(L)$ rounds assuming $h(L)$ is given. We now propose an algorithm to solve lattice agreement which has round complexity related to $h(L)$ even when $h(L)$ is not known. This algorithm called $LA_{\gamma}$ (shown in Figure \ref{fig:la_gamma}), solves the standard lattice agreement in $O(min\{\log^2 h(L), \log^2 f\})$ rounds. The basic idea is to \quotes{guess} the height of $L$ and apply algorithm $LA_{\alpha}$ using the guessed height as input. The algorithm is composed of two phases. At \textit{Phase A}, each process simply broadcasts its value and takes the join of all received values. \textit{Phase B} is the guessing phase which invokes algorithm $LA_{\alpha}$ repeatedly. Notice that \textit{decided} variable is updated at \textit{line 6} of $LA_{\alpha}$.

Let $w_i$ denote the value of $v_i$ after \textit{Phase A}. Let $\Psi$ denote the sublattice formed by values of all correct processes after \textit{Phase A}, i.e, $\Psi = \{u ~| ~(u \in L) \wedge (\exists i: w_i \leq u)\}$. Since there are at most $f$ failures, we have $h(\Psi) \leq f$. Now we show that \textit{Phase B} terminates in at most $\lceil \log h(\Psi) \rceil$ executions of $LA_{\alpha}$. We call the $i$-th execution of $LA_{\alpha}$ as iteration $i$. Notice that the guessed height of iteration $i$ is $2^i$.

\begin{lemma}\label{lem:decide_gamma}
After iteration $\lceil \log h(\Psi) \rceil$ of $LA_{\alpha} $ at \textit{Phase B}, all processes decide.
\begin{proof}
Since $2^{\lceil \log h(\Psi) \rceil} \geq h(\Psi)$, Lemma \ref{lem:dec} still holds which implies Lemma \ref{lem:term}. Thus, all undecided processes have equal values at the last round of iteration $\lceil \log h(\Psi) \rceil$.  Therefore, all undecided processes decide after iteration $\lceil \log h(\Psi) \rceil$.
\end{proof}
\end{lemma}

We now show that two processes decide on comparable values irrespective of whether they both decide on the same iteration of $LA_{\alpha}$.

\begin{lemma}\label{lem:comp_gamma}
($Comparability$) Let $i$ and $j$ be any two processes that decide on value $y_i$ and $y_j$, respectively. Then $y_i$ and $y_j$ are comparable. 
\begin{proof}
Assume process $i$ decides on $G_i$ at round $r_i$ of execution $e_i$ and process $j$ decides on $G_j$ at round $r_j$ of execution $e_j$. If $e_i = e_j$, then $y_i$ and $y_j$ are comparable by Lemma \ref{lem:comp_alpha}. Otherwise, $e_i \neq e_j$. Without loss of generality, suppose $e_i < e_j$. Consider round $r_i$ of execution $e_i$. Since $i$ decides on value $y_i$ at this round, then from Lemma \ref{lem:decideComp}, we have that $y_i$ is comparable with $v_k^r$ for any correct process $k$. Let $V = \{v_k^r ~|~ k ~is ~correct\}$. Then, $y_j$ is at most the join of a subset of $V$. From Lemma \ref{lem:lat}, it follows that $y_i$ is comparable with $y_j$.  
\end{proof}
\end{lemma}

\begin{theorem}\label{theo:comp_gamma}
$LA_{\gamma}$ solves the lattice agreement problem and can tolerate $f < n$ failures.
\begin{proof} 
\textit{Downward-Validity} follows from that fact that the value of each process is non-decreasing along the execution. \textit{Upward-Validity} follows since each process can receive at most all values from all processes. \textit{Comparability} holds by Lemma \ref{lem:comp_gamma}. 
\end{proof}
\end{theorem}
{\bf Complexity}. From Lemma \ref{lem:decide_gamma}, we know \textit{Phase B} terminates in at most $\lceil \log h(\Psi) \rceil$ executions of $LA_{\alpha}$. Thus, \textit{Phase B} takes $\log 2 + \log 4 + ... + \lceil \log h(\Psi) \rceil = \frac{(\lceil \log h(\Psi) \rceil + 1)*( \lceil \log h(\Psi) \rceil)}{2}$ rounds in worst case. Since $h(\Psi) \leq f$ and $h(\Psi) \leq h(L)$, $LA_{\gamma}$ has round complexity of $\min \{O(\log^2 h(L)), O(\log^2 f)\}$. Each process sends $n$ messages at each round, thus message complexity is $n^2 \cdot \min \{O(\log^2 h(L)), O(\log^2 f)\}$.

\section{Lattice Agreement in Asynchronous Systems}
In this section, we discuss the lattice agreement problem in asynchronous systems. The algorithm proposed in \cite{faleiro2012generalized} requires $O(n)$ units of time, whereas our algorithm ($LA_{\delta}$ shown in Figure \ref{fig:la_delta}) requires only $O(f)$ units of time. We first note that

\begin{theorem}\label{theo:impossibility}
The lattice agreement problem cannot be solved in asynchronous message systems if $f \geq \frac{n}{2}$. 
\end{theorem}
\begin{proof}
The proof follows from the standard partition argument. If two partitions have incomparable values then they can never decide on comparable values.
\end{proof}

\subsection{Algorithm $LA_{\delta}$}
On account of Theorem \ref{theo:impossibility}, we assume that $f < \frac{n}{2}$.
The algorithm proceeds in {\em round-trips}. A single round-trip is composed of sending messages to all and getting $n - f$ acknowledgement messages back. At each round-trip, a process sends a \textit{prop} message to all, with its current accepted value as the proposal value, and waits for $n - f$ $ACK$ messages. If majority of these $ACK$ messages are \textit{accept}, then it decides on its current proposed value. Otherwise, it updates its current accept value to be the join of all values received and start next round-trip. Whenever a process receives a proposal, i.e, a \textit{prop} message, if the proposal has a value at least as big as its current value, then it sends back an $ACK$ message with \textit{accept} and updates its current accept value to be the received proposal value. Otherwise, it sends back an $ACK$ message with \textit{reject}. 

\begin{figure}[htb] 
\begin{centering}
\fbox{\begin{minipage}  {2.4in}
\underline{$\bm{LA_{\delta}}$ for $p_i$}\\
\h \textit{acceptVal} := $x_i$// accept value\\
\h $learnedVal := \bot$ // learned value\\
\h \\
 {\bf on receiving} \textit{prop}$(v_j, r)$ from $p_j$:\\
\h {\bf if} {$v_j \geq $ \textit{acceptVal}}\\
\h\h	Send \textit{ACK}(\textit{\quotes{accept},} $-, r$) \\
\h\h \textit{acceptVal} := $v_j$ \\
\h {\bf else}  \\
\h\h Send \textit{ACK}(\textit{\quotes{reject},} \textit{acceptVal}, $r$)\\
\\
\end{minipage}
\hspace{0.2in}
\begin{minipage}  {2.7in}
{\bf for $r := 1$ to $f + 1$} \\
\h \textit{val} := \textit{acceptVal}\\
\h  Send \textit{prop}(\textit{val}, $r$) to all\\
\h  {\bf wait for} $n - f$ \textit{ACK}($-, -, r$) messages \\
\h let $V_r$ be values contained in \textit{reject ACKs} \\
\h let \textit{tally} be number of \textit{accept ACKs} \\
\h {\bf if} {\textit{tally} $> \frac{n}{2}$} \\
\h\h  $learnedVal := val$ \\
\h\h {\bf break} \\
\h {\bf else}\\
\h\h \textit{acceptVal} := \textit{acceptVal} $\sqcup \{v ~| ~v \in V_r\}$ \\
{\bf end for}
\end{minipage}
} 
\end{centering}
\caption{Algorithm $LA_{\delta}$ \label{fig:la_delta}}
\vspace*{-0.2in}
\end{figure}

Let $acceptVal_i^{r}$ denote the accept value (variable \textit{acceptVal}) held by $p_i$ at the beginning of round-trip $r$. Let $L^{(r)} = \{u~|~(u \in L) \wedge (\exists i: acceptVal_i^{r} \leq u)\}$, i.e, $L^{(r)}$ denotes the join-closed subset of $L$ that includes the accept values held by all undecided processes at the beginning of the round-trip $r$. Notice that $L^{(1)} = L$.
\begin{lemma}\label{lem:height_decrease}
 For any round-trip $r$, $h(L^{(r + 1)}) <  h(L^{(r)})$. 
\begin{proof}
If a process decides at round-trip $r$, its value is not in $L^{(r + 1)}$. So, we only need to prove that $h(acceptVal_i^r) < h(acceptVal_i^{r + 1})$ for any process $i$ which does not decide at round-trip $r$. The fact that process $i$ does not decide at round-trip $r$ implies that $i$ must have received at least one \textit{reject} \textit{ACK} with a greater value. Since $acceptVal_i^{r + 1}$ is the join of all values received at round-trip $r$, $acceptVal_i^{r} < acceptVal_i^{r + 1}$. Hence, $h(acceptVal_i^r) < h(acceptVal_i^{r + 1})$ for any undecided process $i$. Therefore, $h(L^{(r)}) < h(L^{(r + 1)})$. 
\end{proof}
\end{lemma}

\begin{lemma}\label{lem:decide_delta}
All process decide within $\min \{h(L), f + 1\}$ asynchronous round-trips.
\begin{proof}
We first show that $h(L^{(2)}) \leq f$. At the first round-trip, each process receives $n - f$ $ACK$s, which is equivalent to receiving $n - f$ values. Therefore, $h(L^{(2)}) \leq f$. Let $r_{min} = \min \{h(L), f + 1\}$. Combining the fact that $h(L^{(2)}) \leq f$ with Lemma \ref{lem:height_decrease}, we have $h(L^{(r_{min})}) \leq 1$. This means that undecided correct processes have the same value. Thus, all of them receive $n - f$ $ACK$ messages with \textit{accept} and decide. Therefore, all processes decide within $min\{h(L), f + 1\}$ round-trips.
\end{proof}
\end{lemma}

We note here that the algorithm in \cite{faleiro2012generalized} takes
$O(n)$ message delays for a value to be learned in the worst case. A crucial difference between $LA_{\delta}$ and the algorithm in \cite{faleiro2012generalized} is that $LA_{\delta}$ starts with the accepted value as the input value. Hence, after the first round-trip, there is a significant reduction in the height of the sublattice, from $n$ initially (in the worst case) to $f$. In \cite{faleiro2012generalized}, acceptors start with the accepted value as null. Hence, there is reduction of height by only $1$ in the worst case. Since in their algorithm, acceptors are different from proposers (in the style of Paxos), acceptors do not have access to the proposed values.

\begin{theorem}\label{theo:comp_delta}
Algorithm $LA_{\delta}$ solves the lattice agreement problem in $\min \{h(L), f + 1\}$ round-trips.
\begin{proof} 
\textit{Down-Validity} holds since the accept value is non-decreasing for any process $i$. \textit{Upward-Validity} follows because each learned value must be the join of a subset of all initial values which is at most $\sqcup \{x_1,...,x_n\}$. For \textit{Comparability}, suppose process $i$ and $j$ decide on values $y_i$ and $y_j$. There must be at least one process that has accepted both $y_i$ and $y_j$. Since each process can only {\em accept} comparable values. Thus, we have either $y_i \leq y_j$ or $y_j \leq y_i$.
\end{proof}
\end{theorem}

{\bf Complexity}. From Lemma \ref{lem:decide_delta}, we know that $LA_{\delta}$ takes at most $\min \{h(L), f + 1\}$ round-trips, which results in $2 \cdot \min \{h(L), f + 1\}$ message delays, since one round-trip takes two message delays. At each round-trip, each process sends out at most $2n$ messages. Thus, the number of messages for all processes is at most $2 \cdot n^2 \cdot \min \{h(L), f + 1\}$. 

\section{Generalized Lattice Agreement}
In this section, we discuss the generalized lattice agreement problem as defined in Section 2.3. Since it is easy to adapt algorithms for lattice agreement in synchronous systems to solve generalized lattice agreement problem, we only consider asynchronous systems. We show how to adapt $LA_{\delta}$ to solve the generalized lattice agreement problem (algorithm $GLA_{\alpha}$ shown in Figure \ref{fig:gla_alpha}) in $\min \{O(h(L)), O(f)\}$ units of time. 

\subsection{Algorithm $GLA_{\alpha}$}
$GLA_{\alpha}$ invokes the Agree() procedure to learn a new value multiple times. The Agree() procedure is an execution of $LA_{\delta}$ with some modifications (to be given later). A sequence number is associated with each execution of the Agree() procedure, thus each correct process has a learned value for each sequence number. The basic idea of $GLA_{\alpha}$ is to let all processes sequentially execute $LA_{\delta}$ to learn values, and make sure: 1) any two learned values for the same sequence number are comparable, 2) any learned value for a bigger sequence number is at least as big as any learned value for a smaller sequence number. The first goal can be simply achieved by invoking $LA_{\delta}$ with the sequence number. In order to achieve the second goal, the key idea is to make any proposal for sequence number $s + 1$ to be at least as big as the largest learned value for sequence number $s$. Notice that at each round-trip of $LA_{\delta}$ execution, a process waits for $n - f$ \textit{ACKs}, and any two set of $n - f$ processes have at least one process in common. Thus, the second goal can be achieved by making sure at least $n - f$ processes know the largest learned value after execution of $LA_{\delta}$ for a sequence number.

Upon receiving a value $v$ from client in a message tagged with \textit{ClientValue}, a process adds $v$ into its buffer and sends a \textit{ServerValue} message with $v$ to all other processes. The process can start to learn new values only when it succeeds at its current proposal. Otherwise, $LA_{\alpha}$ may not terminate, as shown by an example in \cite{faleiro2012generalized}. Upon receiving a \textit{ServerValue} message with value $v$, a process simply adds $v$ to its buffer.
The Agree() procedure is automatically executed when the guard condition is satisfied; that is, it is not currently proposing a value and it has some value in its buffer or it has seen a sequence number bigger than its current sequence number.  
Inside the Agree() procedure, a process first updates its {\em acceptVal} to be the join of current {\em acceptVal} and {\em buffVal}. Then, it starts an adapted $LA_{\delta}$ execution. The original $LA_{\delta}$ and adapted $LA_{\delta}$ differ in the following ways: 1) Each message in the adapted $LA_{\delta}$ is associated with a sequence number. 2) A process can also decide on a value for a sequence number if it receives any \textit{decide} \textit{ACK} message for that sequence number. 3) On receiving a \textit{prop} message associated with a sequence number $s'$, if $s'$ is smaller than its current sequence number which means it has learned a value for $s'$, then it simply sends \textit{ACK} message with its learned value for $s'$ back. If $s'$ is greater than its current sequence number, it updates its \textit{maxSeq} and waits until its current sequence number matches $s'$, then it executes the same procedure as the original $LA_{\delta}$, i.e, it sends back \textit{ACK} message with \textit{accept} or \textit{reject} based on whether the proposal value is bigger than its current accept value or not. The reason why a process keeps track of the maximal sequence number it has ever seen is to make sure each process has a learned value for each sequence number. When the maximum sequence number is 
bigger than
its current sequence number, it has to invoke Agree() procedure even if it does not have any new value to propose. After execution of adapted $LA_{\delta}$, a process increases its current sequence number.
\begin{figure}[htb]\begin{centering}
\fbox{\begin{minipage}  {2.7in}
\underline{$\bm{GLA_{\alpha}}$ for $p_i$}\\
\h \textit{s} := 0 // sequence number\\
\h \textit{maxSeq} := -1 // max seq number seen\\
\h \textit{buffVal} := $\bot$ // received values\\
\h /* map from seq to learned value */\\
\h \textit{LV} := $\bot$ \\
\h \textit{acceptVal} := $\bot$ \\
\h \textit{active} := \textit{false} \\ 

\textbf{on receiving} \textit{ClientValue}($v$):\\
\h \textit{buffVal} := \textit{buffVal} $ \sqcup ~v$\\
\h Send \textit{ServerValue}($v$) to all\\
\\
\textbf{on receiving} \textit{ServerValue}($v$):\\
\h \textit{buffVal} := \textit{buffVal} $ \sqcup ~v$\\
\\
{\bf on receiving} \textit{prop}$(v_j, r, s')$ from $p_j$:\\
\h {\bf if} {$s' < s$}\\
\h\h Send \textit{ACK}(\textit{\quotes{decide}, LV}[$s'$], $r, s'$)\\
\h\h {\bf break}\\
\h {\bf if}  {$s' > s$}\\
\h \h \textit{maxSeq} := $\max \{s', \textit{maxSeq}\}$\\
\h\h {\bf wait until} $s = s'$\\
\h {\bf if} {$v_j \geq$ \textit{acceptVal}}\\
\h\h Send \textit{ACK}(\textit{\quotes{accept},} $-, r, s'$) \\
\h\h \textit{acceptVal} $:= v_j$ \\
\h {\bf else}  \\
\h\h Send \textit{ACK}(\textit{\quotes{reject}}, \textit{acceptVal}, $r, s'$)\\
\end{minipage}

\hspace{0.1in}
\begin{minipage}  {2.5in}
  \textbf{Procedure Agree():}\\
  {\bf guard:} (\textit{active} = $false$) \\
  \hspace*{0.4in} $\wedge$ (\textit{buffVal} $\neq \bot ~ \vee$ 
  \textit{maxSeq} $\geq s$) \\
  {\bf effect:} \\
\h \textit{active} := $true$\\
\h \textit{acceptVal} := \textit{buffVal} $\sqcup$ \textit{acceptVal} \\
\h \textit{buffVal} := $\bot$\\
\h\\
\h /* $LA_{\delta}$ with sequence number */\\
\h {\bf for $r := 1$ to $f + 1$} \\
\h\h \textit{val} := \textit{acceptVal}\\
\h\h  Send \textit{prop}(\textit{val}, $r, s$) to all\\
\h\h  {\bf wait for} $n - f$ \textit{ACK}($-, -, r, s$)\\
\h\h let $V$ be values in \textit{reject ACKs} \\
\h\h let $D$ be values in \textit{decide ACKs} \\
\h\h let \textit{tally} be number of \textit{accept ACKs} \\
\h\h {\bf if} {$|D| > 0$}\\
\h\h\h $val := \sqcup \{d | d \in D\}$ \\
\h\h\h {\bf break}\\
\h\h {\bf else if} {\textit{tally} $> \frac{n}{2}$} \\
\h\h\h {\bf break}\\
\h\h {\bf else}\\
\h\h\h \textit{acceptVal} := \textit{acceptVal} \\
\hspace*{1.2in} $\sqcup \{v ~| ~v \in V\}$ \\
\h {\bf end for}\\
\h \textit{LV}$[s]$ := \textit{val} \\
\h \textit{s} := \textit{s} + 1\\
\h \textit{active} := \textit{false}\\
\end{minipage}
} 
\end{centering}
\caption{Algorithm $GLA_{\alpha}$ \label{fig:gla_alpha}}
\vspace*{-0.2in}
\end{figure}

We next show the correctness of $GLA_{\alpha}$. Let $acceptVal_s^{p}$ denote the \textit{acceptVal} of process $p$ at the end of Agree() procedure for sequence number $s$. Let $LV_p$ denote the map of sequence number to learned value (variable $LV$) for process $p$ and $m_s = \sqcup \{LV_p[s]: p \in [1..n]\} $, i.e, $m_s$ denotes the join of all learned values for sequence number $s$. Let $LP_s = \{p~ |~ (p \in [1..n]) \wedge (m_s \leq acceptVal_s^{p}) \}$, i.e, $LP_s$ is the set of processes which have \textit{acceptVal} greater than the join of all learned values for the sequence number $s$. Notice that a process has two ways to learn a value for its current sequence number in the Agree() procedure: 1) by receiving a majority of \textit{accept} \textit{ACKs}. 2) by receiving some \textit{decide} \textit{ACKs}. 

The following lemma proves that the adapted $LA_{\alpha}$ satisfies the first goal. 
\begin{lemma}\label{lem:comp_same_sequence}
For any sequence number $s$, \textit{$LV_p[s]$} is comparable with \textit{$LV_q[s]$} for any two processes $p$ and $q$.
\begin{proof}
We only need to show that any two processes which learn by the first way must learn comparable values, since processes which learn by the second way simply learn values from processes which learn by the first way. By the same reasoning as \textit{Comparability} of Theorem \ref{theo:comp_delta}, we know this is true. 
\end{proof}
\end{lemma}

From Lemma \ref{lem:comp_same_sequence}, we know that $m_s$ is the largest learned value for sequence number $s$. 
\begin{lemma}\label{lem:maj_largest_alpha}
For any sequence number $s$, $|LP_s| > \frac{n}{2}$.
\begin{proof}
Consider Agree() procedure for $s$. Since $m_s$ is the largest learned value for sequence number $s$, there must exist a process $p$ which learns $m_s$ by the first way. Thus, $p$ must have received a majority of \textit{accept} \textit{ACKs}, which means at least a majority of processes have \textit{acceptVal} greater than $m_s$ after Agree() procedure for $s$. Therefore, $|LP_s| > \frac{n}{2}$. 
\end{proof}
\end{lemma}

The lemma below shows that $GLA_{\alpha}$ achieves the second goal. 
\begin{lemma}\label{lem:stability_comp}
$m_s \leq LV_p[s + 1]$ for any process $p$ and any sequence number $s$.
\begin{proof}
From Lemma \ref{lem:maj_largest_alpha}, we know for sequence number $s$ at least a majority of processes have \textit{acceptVal} greater than $m_s$. To decide on $LV_p[s + 1]$, process $p$ must get majority \textit{accept}. Since any two majority has at least one process in common, $m_s \leq LV_p[s + 1]$.
\end{proof}
\end{lemma}

\begin{theorem}\label{theo:corretness_gla_alpha}
Algorithm $GLA_{\alpha}$ solves generalized lattice agreement when a majority of processes is correct.
\begin{proof}
\textit{Validity} holds since any learned value is the join of a subset of values received.   

\textit{Stability}. 
From Lemma \ref{lem:stability_comp} and the fact that $LV_p[s] \leq m_s$, we have that $LV_p[s] \leq LV_p[s + 1]$ for any process $p$ and any sequence number $s$, which implies {\em Stability}.

\textit{Comparability}. We need to show that \textit{$LV_p[s]$} and \textit{$LV_q[s']$} are comparable for any two processes $p$ and $q$, and for any two sequence number $s$ and $s'$. If $s = s'$, this is immediate from Lemma \ref{lem:comp_same_sequence}. Now consider the case when $s \neq s'$. Without loss of generality, assume $s < s'$. From Lemma \ref{lem:stability_comp}, we can conclude that \textit{$LV_p[s]$} $ \leq $ \textit{$LV_q[s']$}. Thus, {\em comparability} holds.

\textit{Liveness}. Any received value $v$ is eventually included in some proposal, i.e, \textit{prop} message. From Theorem \ref{theo:comp_delta}, we know that in at most $2 \cdot \min \{h(L), f + 1\}$ message delays that proposal value will be included in some learned value. Thus, $v$ will be learned eventually.
\end{proof}
\end{theorem}
{\bf Complexity.} For time complexity, from the analysis for \textit{liveness} in Theorem \ref{theo:corretness_gla_alpha}, we know that a received value is learned in at most $2 \cdot \min \{h(L), f + 1\}$ message delays. For message complexity, since each process sends out $n$ messages per \textit{round-trip}, the total number of messages needed to learn a value is $2 \cdot n^2 \cdot \min \{h(L), f + 1\}$.

\section{Conclusions}
We have presented algorithms for the lattice agreement problem and the generalized lattice agreement problem. These algorithms achieve significantly better time complexity than previous algorithms. For future work, we would like to know the answers to the following two questions: 1) Is $\log f$ rounds the lower bound for lattice agreement in synchronous message passing systems? 2) Is $O(f)$ message delays optimal for the lattice agreement and generalized lattice agreement problem in asynchronous message passing systems?
\subparagraph*{Acknowledgements.}

We want to thank John Kaippallimalil for providing some useful application cases for CRDT and generalized lattice agreement.

\appendix


\bibliography{ref}

\end{document}